\documentclass[a4paper,twoside]{amsart}
\usepackage{amsmath}
\usepackage{amsfonts}
\usepackage{amssymb}
\usepackage{amsthm}
\usepackage{newlfont}
\usepackage{graphicx}
\usepackage{amscd}
\usepackage{young}

\theoremstyle{plain}
\newtheorem{Th}{Theorem}[section]
\newtheorem{Cor}[Th]{Corollary}

\newtheorem{Prop}[Th]{Proposition}
\theoremstyle{definition}

\newtheorem{Ex}{Example}[section]
\theoremstyle{remark}
\newtheorem*{Rem}{Remark}
\numberwithin{equation}{section}

\newcommand{\ZZ}{{\mathbb Z}}

\newcommand{\bphi}{\boldsymbol{\phi}}

\begin{document}

\title[The Coxeter relations and KP map for non-commuting symbols]
{The Coxeter relations and KP map  for non-commuting symbols}

\author{Adam Doliwa}

\address{A. Doliwa: Faculty of Mathematics and Computer Science, University of Warmia and Mazury in Olsztyn,
ul.~S{\l}oneczna~54, 10-710~Olsztyn, Poland} 
\email{doliwa@matman.uwm.edu.pl}
\urladdr{http://wmii.uwm.edu.pl/~doliwa/}

\author{Masatoshi Noumi}
\address{M. Noumi: Department of Mathematics, Kobe University, Rokko, Kobe 657-8501, Japan}
\email{noumi@math.kobe-u.ac.jp}

\date{}

\begin{abstract}
We give an action of the symmetric group on non-commuting indeterminates in terms of series in the corresponding Mal'cev--Newmann division ring. The action is constructed from  the non-Abelian Hirota--Miwa (discrete KP) system. The corresponding companion map, which gives generators of the action, is discussed in the generic case and the corresponding explicit formulas have been found in the periodic reduction. We discuss also briefly connection of the companion to the KP map with context-free languages.
\end{abstract}
\maketitle

\section{Introduction}
Non-commutative extensions of integrable systems are of growing interest in mathematical physics~\cite{FWN-Capel,Kupershmidt,BobSur-nc,Nimmo-NCKP,EtingofGelfandRetakh,DF-K,Konstantinou-Rizos,GrYB-GKRM,Dol-Des,Dol-GD,Dol-qP6}. They may be considered as a useful platform to more thorough understanding of integrable quantum or statistical mechanics lattice systems, where the quantum Yang--Baxter equation~\cite{Baxter,QISM} plays a role. Recently, integrable maps with variables satisfying various commutation relations (including the canonical Weyl relations and those coming from the quantum group theory) have been studied in \cite{Hasegawa,Kuroki,Sergeev-q3w,DoliwaSergeev-pentagon,BazhanovSergeev,Tsuboi}. 

Such a direction from commutative to non-commutative objects in the theory of integrable systems is not an isolated issue. One can say it is in mainstream of current developments in contemporary mathematics. The origins of this trend are rooted back in quantum physics, where the idea of replacing functions by not necessarily commuting operators appeared. There is a topology of non-commutative spaces, a non-commutative integration theory, a non-commutative differential geometry \cite{Connes}, non-commutative probability theory~\cite{FreeProbability}, etc. The present research aims to make a step in direction to create a kind of free integrability. The advantage of working with totally non-commutative variables has been advocated also in~\cite{Quasideterminants}. As the geometric research \cite{Dol-Des} suggests, the algebraic environment of our work will be the theory of Mal'cev--Neumann division rings constructed from free groups.

In this paper we study maps obeying the Coxeter relations, which are obtained from the non-commutative discrete KP system of equations (called originally in \cite{Nimmo-NCKP} the non-Abelian Hirota--Miwa system). In the commutative case such maps were studied in~\cite{KNY-A,Etingof,SurisVeselov}, see also~\cite{Doliwa-YB} for a version with certain commutativity restrictions.  We remark that  Hirota's discrete KP equation~\cite{Hirota-1983} gives as reductions majority of known integrable systems~\cite{KNS-rev}, both discrete and continuous.
The $A$-type Coxeter relations~\cite{Humphreys} considered in the paper are closely related to the notion of Yang--Baxter maps (originally \cite{Drinfeld} called the set-theoretical solutions of quantum Yang--Baxter equation). We will use, well established by now \cite{ABS-YB},
correspondence between the Yang--Baxter maps and the multidimensional consistency property of integrable equations on quad-graphs \cite{ABS,FWN-cons}.

The paper is organized as follows. In the next section we first collect relevant results about the Mal'cev--Neumann construction of division rings. Then we discuss the Coxeter relations in connection to the multidimensional consistency. We close the preparatory Section~\ref{sec:prelim} with discussion of properties of a map arising from the non-Abelian Hirota--Miwa (discrete KP) system. In Section \ref{sec:KP-comp-map} we study the companion map to the above one. The corresponding explicit formulas in the case of periodic reduction, given in Section~\ref{sec:KP-series}, generalize to fully non-commutative level the maps studied in~\cite{KNY-A,Doliwa-YB}.  We discuss also in Section~\ref{sec:CF} relation of the series representing the maps to context-free languages. A surprising consequence is that, although the original non-commutative KP map is given by a rational expression, the series which give the companion map are not rational. 

\section{Preliminaries} \label{sec:prelim}
\subsection{The Mal'cev--Neumann construction and rational series over free group}
Let $\Bbbk$ be a field (in our case the field of rational numbers $\mathbb{Q}$), and let $G$ be a group. Consider the $\Bbbk$-vector space $\Bbbk[[G]]$ of formal series over $G$ with coefficients in $\Bbbk$, where a series $S\in\Bbbk[[G]]$ is just a map  $G\ni g \mapsto (S,g)\in \Bbbk$ represented by
\begin{equation*}
S = \sum_{g\in G} (S,g) g .
\end{equation*} 
The support of the series $S$ is then the set $\mathrm{supp}(S) = \{g \in G \, | \, (S,g) \neq 0\}$. Series with finite support can be multiplied using the formula
\begin{equation} \label{eq:Cauchy-prod}
(ST,h) = \sum_{fg=h} (S,f) (T,g), 
\end{equation} 
and form the group algebra $\Bbbk[G]$.

Assume that $G$ is totally ordered by the order $<$ compatible with its group structure, i.e. 
\begin{equation*}
\forall f,g,h\in G, \qquad f < g \quad \Rightarrow \quad  \begin{cases}
fh < gh .\\
hf < hg .
\end{cases}
\end{equation*}
Recall that $ P \subset G$ is well ordered if every non-empty subset of $P$
admits a smallest element for the order of $G$. The set $\Bbbk((G,<))$ of Mal'cev--Neumann series it is the set of the series in $\Bbbk[[G]]$ whose supports
are well ordered (with respect to the given order $<$). One can equip $\Bbbk((G,<))$ with a $\Bbbk$-algebra structure by defining the (Cauchy) product of two Mal'cev--Neumann  series $S$ and $T$ by \eqref{eq:Cauchy-prod}.
The remarkable result by Mal'cev~\cite{Malcev}, and independently by Neumann~\cite{Neumann}, is that this algebra is a division ring (skew field).
Moreover, the space of Mal'cev--Neumann  series is complete in the natural ultra-metric topology \cite{Neumann}, where a sequence converges if and only if the coefficient of every monomial stabilizes.

\begin{Ex} \label{ex:Coxeter-QN}
	Let $G=\mathbb{Z}$ with its standard order, then 
	\begin{equation*}
	\Bbbk((\mathbb{Z},<)) = \left\{ \sum_{n=d}^\infty a_n x^n \, | \, a_n\in \Bbbk, \; d\in\mathbb{Z} \right\},
	\end{equation*}
	thus, we obtain the usual formal Laurent series (bounded from below) in one variable over $\Bbbk$. If we took the
	opposite order on $\mathbb{Z}$, we would obtain the Laurent series in $x^{-1}$ over $\Bbbk$.
\end{Ex}

The free group $\Gamma = \Gamma (x_1,x_2,\dots)$  with generators $x_1,x_2, \dots$, possibly an infinite list, can be regarded as the set of equivalence classes of finite words in the letters $x_i$ and their formal inverses $x_i^{-1}$, where the words are considered equivalent if one can pass from one to the other by removing (or inserting) consecutive letters of the form $x_ix_i^{-1}$ or $x_i^{-1}x_i$. The group operation is concatenation and the empty word represents the identity element. By fixing a total order $<$ in $\Gamma$ compatible with its group structure one can consider the corresponding  Mal'cev--Neumann division ring $\Bbbk((\Gamma,<))$.

There are number of ways to order free groups \cite{Vinogradov} --- in fact for free groups with more than one generator there are uncountably many. The standard one, described in \cite{Bergman} in terms of the Magnus embedding~\cite{Magnus} is as follows. 
Define the ring
$$\Lambda = \mathbb{Z}[[X_1,X_2,\dots ]]$$
to be the ring of formal power series in the non-commuting variables $X_i$, one for each generator of $\Gamma$. If there are infinitely many variables, we only allow for expressions involving a finite set of variables to belong to $\Lambda$, so that an element of $\Lambda$ has only a finite number of terms of a given degree. The advantage of $\Lambda$ is that the variables have no negative exponents (unlike in $\Gamma$) and so it is easier to define an ordering, without having to worry about cancellation problems. Define the (multiplicative) homomorphism $\mu\colon\Gamma\to\Lambda$ which on the generators of $\Gamma$ is given by
\begin{align*}	
\mu(x_i) & = 1 + X_i \, ,\\
\mu(x_i^{-1}) & = 1 - X_i + X_i^2 - X_i^3 + \cdots . 
\end{align*}
To compare two elements in $\Lambda$ 
write them with terms of lower degree first, and lexicographically order terms of the same degree (this is the so called shortlex order). Then the two elements are ordered according to the coefficient of the first term at which they differ. Finally, define the ordering on the free group $\Gamma$ by declaring
\begin{equation*}
v < w \quad \text{in} \quad \Gamma \Leftrightarrow \; \mu(v) < \mu(w) \quad \text{in} \quad \Lambda .
\end{equation*}
\begin{Ex}
	To compare $x_1x_2^2x_1^{-1}$ and $x_1^{-1}x_2^2x_1$ notice that the series
	\begin{align*}
	\mu(x_1x_2^2x_1^{-1}) & = 1 + 2X_2 + 2X_1 X_2 - 2X_2 X_1 + X_2^2 + \cdots \\
	\mu(x_1^{-1}x_2^2x_1) & = 1 + 2X_2 - 2X_1 X_2 + 2X_2 X_1 + X_2^2 + \cdots
	\end{align*} 
	have different coefficients first at $X_1X_2$. Because $-2 < 2$ in $\mathbb{Z}$ then we have $x_1^{-1}x_2^2x_1 < x_1x_2^2x_1^{-1}$ .
\end{Ex}

Consider the ordered free group $(\Gamma,<)$ with finite number of generators.
The smallest sub-division ring in $\Bbbk((\Gamma,<))$ containing the group algebra $\Bbbk[\Gamma]$ is called the division ring of fractions. It turns out that the result is independent of the particular order $<$ used in the construction, in the sense that it is isomorphic \cite{Lewin} to the  universal field of fractions,
called also \cite{Cohn} the free skew field.

\subsection{Coxeter relations and multidimensional consistency} \label{sec:Coxeter}
Consider the group given by generators  $\langle r_1, r_2 , \dots , r_{N-1} \rangle $ subject to the following relations (of type $A_{N-1}$), see \cite{Humphreys} for details
\begin{align} \label{eq:C-inv}
r_i\circ r_i & = \mathrm{id} \qquad \text{(involutivity)},\\
r_{i}\circ r_{i+1} \circ r_{i} & = r_{i+1}\circ r_{i}\circ r_{i+1}, 
\qquad \text{(braid relations)}, 
\label{eq:C-braid} \\ \label{eq:C-ind}
r_i \circ r_j & = r_j \circ r_i \quad |i-j| > 1. 
\end{align}
It is well known that such a group is isomorphic to the symmetric group $\mathcal{S}_{N}$, and in this interpretation $r_i$ can be identified with the transposition $(i,i+1)$. The standard geometric realization of the group is the symmetry group of the regular $N$-simplex, where $r_i$ can be identified with the reflection exchanging vertices $i$ with $i+1$.

Let us present another geometric realization of the symmetric group, which will be used in the sequel. 
Consider the $N$-dimensional cube graph $Q_N$, whose vertices can be identified with binary sequences of length $N$, and two vertices belong to an edge if their sequences differ at one place only. The edges can be given an orientation to the vertex with bigger number of ones, and the corresponding index $i=1,\dots,N$ indicating at which position of the sequence the transition $0\to 1$ occurs. This way any oriented path from the initial vertex $(0,0,\dots,0)$ to the terminal vertex $(1,1,\dots,1)$ of the graph $Q_N$ gives rise to a sequence $(\sigma(1),\sigma(2),\dots,\sigma(N))$ of indices of subsequent edges, and thus to the permutation $\sigma\in\mathcal{S}_N$. The left regular action of $\mathcal{S}_N$ on itself gives the action on the paths. In what follows we attach to each edge a weight, and we extend the action of the generators $r_i$ of the symmetric group to the corresponding action on weighted paths.

Let $\mathcal{X}$ be any set, consider map $f\colon \mathcal{X} \times \mathcal{X} \ni (\boldsymbol{x}_1,\boldsymbol{x}_2) \mapsto 
(f_1(\boldsymbol{x}_1,\boldsymbol{x}_2), f_2(\boldsymbol{x}_1,\boldsymbol{x}_2)) \in  \mathcal{X} \times \mathcal{X}$. We would like to interpret the map as providing shifts in discrete variables for fields attached to the edges of a square lattice, so that $\boldsymbol{x}_{1(2)} = f_1(\boldsymbol{x}_1,\boldsymbol{x}_2)$, $\boldsymbol{x}_{2(1)} = f_2(\boldsymbol{x}_1,\boldsymbol{x}_2)$ as graphically depicted in Figure~\ref{fig:3D-GD-x}. In order to keep such an interpretation for multidimensional lattices, the map should satisfy compatibility conditions equivalent to the commutation of shifts in different directions.
\begin{figure}
	\begin{center}
		\includegraphics[width=12cm]{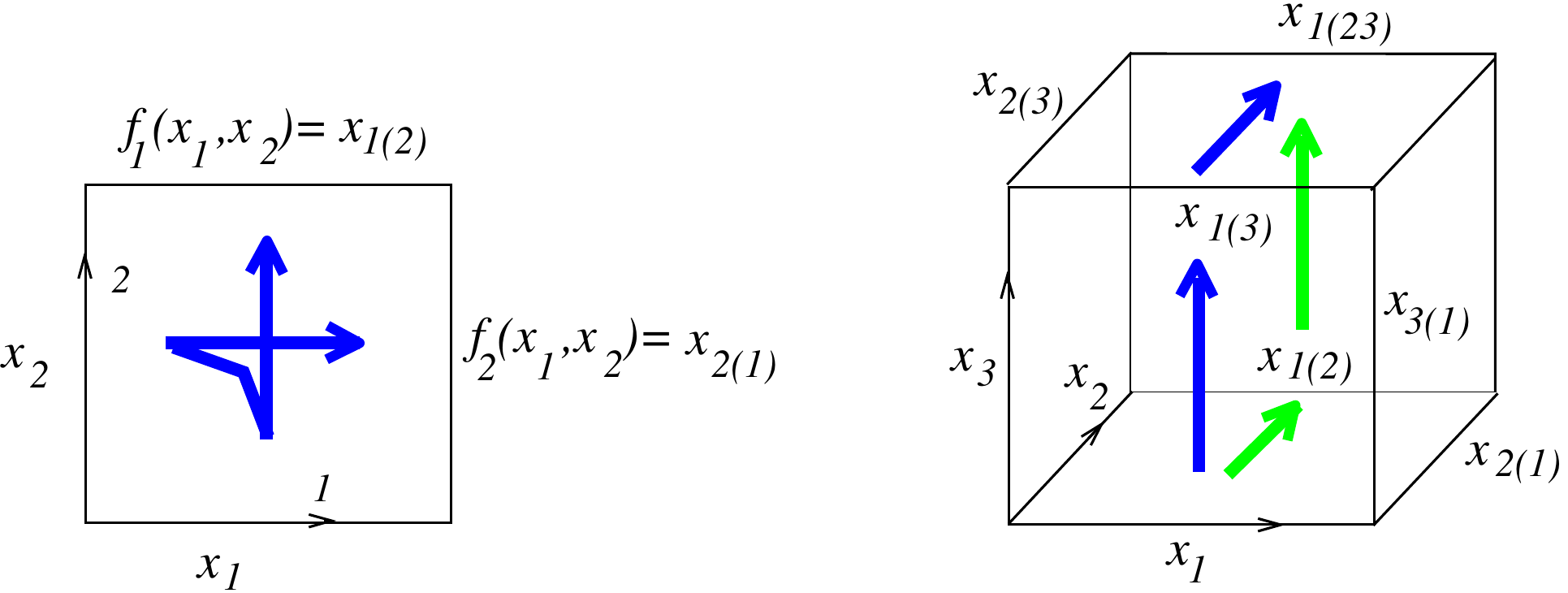}
	\end{center}
	\caption{An edge-map and its consistency around the cube} 
	\label{fig:3D-GD-x}
\end{figure}
 
The map $f$ is called consistent around the cube~\cite{ABS,FWN-cons} if for arbitrary values of the fields $(\boldsymbol{x}_1,\boldsymbol{x}_2,\boldsymbol{x}_3)$ attached to initial edges (containing the vertex $(0,0,0)$) of the cube $Q_3$ 
\begin{align}
f_1(f_1(\boldsymbol{x}_1,\boldsymbol{x}_3), f_1(\boldsymbol{x}_2,\boldsymbol{x}_3) )& = f_1(f_1(\boldsymbol{x}_1,\boldsymbol{x}_2), f_2(\boldsymbol{x}_2,\boldsymbol{x}_3) ) \stackrel{\mathrm{def}}{=} \boldsymbol{x}_{1(23)},\\
f_1(f_2(\boldsymbol{x}_1,\boldsymbol{x}_2), f_2(\boldsymbol{x}_1,\boldsymbol{x}_3) )& = f_2(f_1(\boldsymbol{x}_1,\boldsymbol{x}_3), f_1(\boldsymbol{x}_2,\boldsymbol{x}_3) )
\stackrel{\mathrm{def}}{=} \boldsymbol{x}_{2(13)} ,\\
f_2(f_2(\boldsymbol{x}_1,\boldsymbol{x}_2), f_2(\boldsymbol{x}_1,\boldsymbol{x}_3) )& = f_2(f_1(\boldsymbol{x}_1,\boldsymbol{x}_2), f_2(\boldsymbol{x}_2,\boldsymbol{x}_3) )
\stackrel{\mathrm{def}}{=} \boldsymbol{x}_{3(12)},
\end{align}
i.e. for any final edge (containing the vertex $(1,1,1)$) of the cube two possible values of the corresponding field calculated using the map coincide, see Figure~\ref{fig:3D-GD-x}. This concept is the crucial one in the theory of discrete integrable systems. In particular, it implies possibility of attaching in a consistent way values of the field on all edges of the multidimensional cube $Q_N$ once its values on initial edges have been given. 

Assume that multidimensionally consistent map $f$ is invertible and admits existence of its unique companion map $r$ defined by
\begin{equation}
r(\boldsymbol{x}, \boldsymbol{y}) = (\boldsymbol{x}^\prime, \boldsymbol{y}^\prime), \qquad \text{where} \qquad (\boldsymbol{y}^\prime, \boldsymbol{y}) = f (\boldsymbol{x}, \boldsymbol{x}^\prime).
\end{equation}
It is known~\cite{ABS-YB}, that the multidimensional consistency in that case implies braid relations for the companion map in the form 
\begin{equation} \label{eq:braid-r}
r_{2} \circ r_{1} \circ r_{2} = r_{1} \circ r_{2} \circ r_{1}, 
 \qquad \text{in} \qquad \mathcal{X} \times \mathcal{X}  \times \mathcal{X}  ,
\end{equation}
where  $r_1 = r \times \mathrm{id}_{\mathcal{X}}$ and $r_2 = \mathrm{id}_{\mathcal{X}} \times r$.
\begin{figure}
	\begin{center}
		\includegraphics[width=14cm]{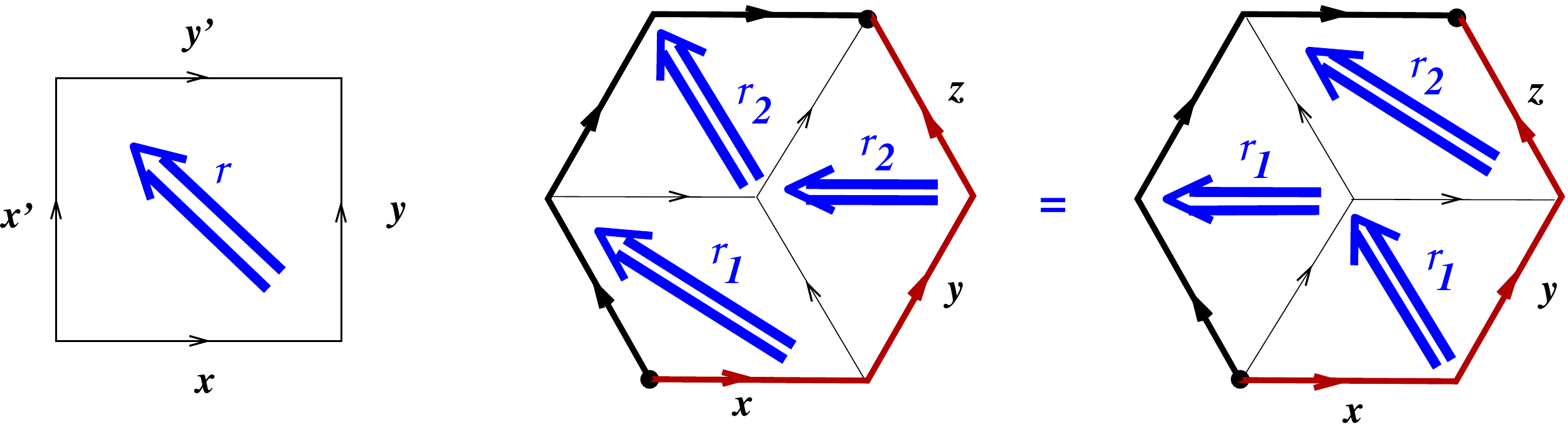}
	\end{center}
	\caption{The companion map and braid relation $r_{2} \circ r_{1} \circ r_{2} = r_{1} \circ r_{2} \circ r_{1}$ in the edge interpretation} 
	\label{fig:2D-EDGE-B}
\end{figure}
Actually, the connection has been formulated in~\cite{ABS-YB} in terms of the 
the Yang--Baxter relations
 \begin{equation}
 R_{12}\circ R_{13} \circ R_{23} = R_{23}\circ R_{13}\circ R_{12}  \qquad \text{in} \qquad \mathcal{X} \times \mathcal{X}  \times \mathcal{X} ,
 \end{equation}
where $R=\pi \circ r$, $\pi$ is the transposition (i.e. ($R(\boldsymbol{x}, \boldsymbol{y}) = (\boldsymbol{y}^\prime, \boldsymbol{x}^\prime)$), and subscripts denote the arguments on which $R$ acts. We remark that eventual involutivity of $r$ means the so called reversibility condition
 \begin{equation}
 \pi \circ R \circ \pi \circ R = \mathrm{id}_{\mathcal{X} \times \mathcal{X} } ,
 \end{equation}
 on the level of the corresponding Yang--Baxter map.
 
\subsection{The non-commutative KP hierarchy and the KP map}
The linear problem for the non-commutative difference KP hierarchy \cite{Doliwa-YB} ( see also \cite{KNY-qKP} for the commutative case) reads as follows
\begin{equation} \label{eq:lin-lKPh}
\bphi_{k+1} - \bphi_{k(i)} = \bphi_k u_{i,k}, 
\end{equation}
here $\bphi_{k}(n_1,n_2, \dots )$, is $k$-th component, $k\in\ZZ$, of the wave function with values in a division ring, $n_i\in\ZZ$ are independent discrete variables, and index in brackets denotes the forward shift in the corresponding variable, i.e.
$\bphi_{k(i)}(n_1,n_2, \dots , n_i , \dots ) = \bphi_{k}(n_1,n_2, \dots , n_i +1 ,\dots )$. The compatibility condition of equations \eqref{eq:lin-lKPh} gives the so called non-commutative difference KP hierarchy
\begin{align} \label{eq:ncKP-1}
& u_{j,k}u_{i,k(j)} = u_{i,k} u_{j,k(i)}, \qquad i\neq j,\\
& u_{i,k(j)} + u_{j,k+1} = u_{j,k(i)} + u_{i,k+1}. \label{eq:ncKP-2}
\end{align}
It is not difficult to find explicit form of the transformation rule
\begin{equation} \label{eq:KP-u-solved}
u_{i,k(j)} = ( u_{i,k} - u_{j,k})^{-1} u_{i,k} ( u_{i,k+1} - u_{j,k+1}), \qquad i\neq j, 
\end{equation} 
called in \cite{Doliwa-YB} the non-commutative KP map. Let us attach the vector $\boldsymbol{u}_i = (u_{i,k})_{k\in\ZZ}$ to the corresponding $i$-th edge of multidimensional cube (see Figure~\ref{fig:3D-GD-x}). Then the map involves edges of a single quadrilateral of the cube. Our paper is based on the following result~\cite{Dol-Des-red} equivalent to integrability of the non-commutative KP hierarchy 
\begin{Prop}
The KP map is multidimensionally consistent.
\end{Prop}
\begin{Rem}
		Equations of the non-commutative discrete KP hierarchy can be obtained from the  non-Abelian Hirota--Miwa system~\cite{Nimmo-NCKP}, see also \cite{Dol-Des} for geometric meaning of the corresponding linear problem and of the nonlinear equations in terms of the so called Desargues maps.
	Notice that the KP map can be considered as two-dimensional only on the formal level of the infinite-component vector $\boldsymbol{u}_i$. On the local level of the three-dimensional non-Abelian Hirota--Miwa system its four dimensional consistency was established in~\cite{Dol-Des} in connection to the Desargues theorem of projective geometry. Then, instead of the Yang--Baxter equation, one has the so called pentagon equation~\cite{DoliwaSergeev-pentagon}. 
\end{Rem}

\begin{Cor} \label{cor:KP-u-inverse}
The KP map is birational with the inverse given by
\begin{equation} \label{eq:KP-u-inverse}
u_{i,k} = ( u_{i,k-1(j)} - u_{j,k-1(i)}) u_{i,k(j)} ( u_{i,k(j)} - u_{j,k(i)})^{-1}.
\end{equation}
\end{Cor}
By direct calculations, along lines described in \cite{Dol-Des-red} for the KP map, one can show the following result.
\begin{Cor}
The inverse of the KP map is three-dimensionally consistent as well, i.e. given values $(\boldsymbol{u}_{i(jl)}, \boldsymbol{u}_{j(il)},\boldsymbol{u}_{l(ij)})$ on three final edges of the cube $Q_3$, the various ways to obtain values $(\boldsymbol{u}_{i}, \boldsymbol{u}_{j},\boldsymbol{u}_{l})$ on the initial edges give the same result.
\end{Cor}

\section{The companion of the KP map and its properties} \label{sec:KP-comp-map}
Define the companion map to the KP map
$r \colon (\boldsymbol{x}, \boldsymbol{y}) \mapsto 
(\boldsymbol{x}^\prime, \boldsymbol{y}^\prime)$
by solution of the equations
\begin{equation} \label{eq:xy-k}
x_k y_k = x_k^\prime y_k^\prime, \qquad y_k + x_{k+1} = y_k^\prime + x_{k+1}^\prime, \qquad k\in\ZZ,
\end{equation}
which is just system \eqref{eq:ncKP-1}-\eqref{eq:ncKP-2} written in new variables 
\begin{equation}
x_k = u_{i,k}, \quad y_k = u_{j,k(i)}, \quad x^\prime_k = u_{j,k}, \quad y^\prime_k = u_{i,k(j)}
\end{equation} 
for fixed pair of distinct indices $i,j$. We assume that the arguments $(\boldsymbol{x}, \boldsymbol{y}) = (x_k,y_k)_{k\in\ZZ}$ of the map  are free variables.

Notice that equation \eqref{eq:xy-k} does not determine the map uniquely. Given one component of the solution, $x_0^\prime$ say, all other components can be obtained then by the two-fold recursion
\begin{align}
y_k^\prime  = (x_k^\prime)^{-1} x_k y_k, \qquad  & x_{k+1}^\prime = y_k + x_{k+1} - y_k^\prime ,  & k\geq 0, \\
y_{k-1}^\prime = y_{k-1} + x_k - x_k^\prime ,  \qquad & x_{k-1}^\prime = x_{k-1}y_{k-1}(y_{k-1}^\prime)^{-1} ,  & k \leq 0 \, .
\end{align}
One can consider then $x_0^\prime$ as free parameter (notice that the assumption $x_0^\prime = x_0$ leads to the identity mapping). 
When we are interested in $K$-periodic solutions to the companion map equations~\eqref{eq:xy-k}, i.e. $k\in\ZZ_K$, then the freedom vanishes. 
In this Section we study properties of the companion map (with free parameter)  which hold even without the periodicity restriction.

	Define recursively polynomials
	\begin{equation} \label{eq:P-k-l-ind}
	\mathcal{P}_{k}^{(0)} = 1, \qquad
	\mathcal{P}_{k}^{(\ell)} = \mathcal{P}_{k}^{(\ell - 1)} \, x_{k+\ell} + 
	y_k \hdots y_{k+\ell-1}  = 
	x_{k+1} \hdots x_{k+\ell}+ y_{k} \mathcal{P}_{k+1}^{(\ell -1)} .
	\end{equation}
	In \cite{Doliwa-YB} it was shown that the polynomials are invariants of the companion map. Following \cite{NY-RSK} we can represent the polynomials as formal sums of all weighted paths between two vertices of the graph as on Figure~\ref{fig:Poly}.
This observation can be also given meaning in terms of matrices.
\begin{Prop}
	The polynomial $\mathcal{P}_{k}^{(\ell)}$ is $(2,1)$ element of the following product of $2\times 2$ matrices
	\begin{equation}
	\left( \begin{array}{cc} x_k & 0 \\ 1 & y_k \end{array} \right)
	\left( \begin{array}{cc} x_{k+1} & 0 \\ 1 & y_{k+1} \end{array} \right) \hdots
	 \left( \begin{array}{cc} x_{k+\ell} & 0 \\ 1 & y_{k+\ell} \end{array} \right) =
	 \left( \begin{array}{cc} x_k \hdots x_{k+\ell} & 0 \\ 
	 \mathcal{P}_{k}^{(\ell)} & y_k \hdots y_{k+\ell} \end{array} \right) 
	\end{equation}
\end{Prop}
\begin{proof}
	By mathematical induction using the recursive definition \eqref{eq:P-k-l-ind}. The combinatorial interpretation of weighted paths in graphs in terms of matrices is standard~\cite{Zeilberger,Stanley-2}.
\end{proof}
\begin{figure}
	\begin{center}
		\includegraphics[width=8cm]{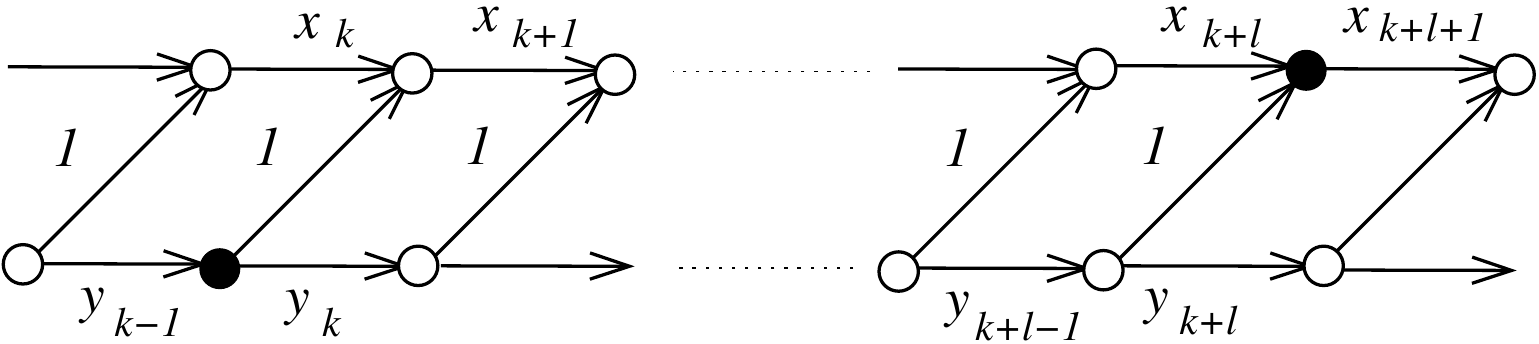}
	\end{center}
	\caption{The combinatorial interpretation of the polynomial $ \mathcal{P}_{k}^{(\ell)}$ as the sum over all weighted paths between two vertices (marked in black)} 
	\label{fig:Poly}
\end{figure}
\begin{Cor} \label{cor-alpha-beta}
Define functions $\alpha_k(\boldsymbol{x},\boldsymbol{y}), \beta_k(\boldsymbol{x},\boldsymbol{y})$, $k\in\mathbb{Z}$, by 
\begin{align} \label{eq:xy-alpha}
 \alpha_k(\boldsymbol{x},\boldsymbol{y})  & = x_k y_k    ,\\
  \beta_k(\boldsymbol{x},\boldsymbol{y}) & =  y_{k} + x_ {k+1}, \label{eq:xy-beta}
\end{align}
or using infinite matrices
\begin{equation*}
\left(  \begin{array}{ccccc}  
\ddots & \ddots  & \ddots & \ddots &  \\
\ddots    & x_{k-1} & 0 & 0 & \ddots \\
\ddots  &	1 & x_{k}  & 0 &  \ddots  \\
\ddots & 0 & 1 & x_{k+1} &  \ddots   \\
& \ddots & \ddots & \ddots & \ddots 
\end{array} \right) 
\left(  \begin{array}{ccccc}  
\ddots & \ddots  & \ddots & \ddots &  \\
\ddots    & y_{k-1} & 0 & 0 & \ddots \\
\ddots  &	1 & y_{k}  & 0 &  \ddots  \\
\ddots & 0 & 1 & y_{k+1} &  \ddots   \\
& \ddots & \ddots & \ddots & \ddots 
\end{array} \right) =
\left(  \begin{array}{ccccc}  
\ddots & \ddots  & \ddots & \ddots &  \\
\ddots    & \alpha_{k-1} & 0 & 0 & \ddots \\
\ddots  &	\beta_{k-1} & \alpha_{k}  & 0 &  \ddots  \\
\ddots & 1 & \beta_{k} & \alpha_{k+1} &  \ddots   \\
& \ddots & \ddots & \ddots & \ddots 
\end{array} \right) .
\end{equation*}	
Then the problem of finding the companion to the KP map \eqref{eq:xy-k} can be formulated as the problem of finding another solution $x_k^\prime, y_k^\prime$ of the system \eqref{eq:xy-alpha}-\eqref{eq:xy-beta} with the data $\alpha_k, \beta_k$, once a solution $x_k, y_k$ has been given or, equivalently, as a matrix refactorization problem.
\end{Cor} 
\begin{Rem}
	It is immediate to see that the functions $\alpha_k(\boldsymbol{x},\boldsymbol{y}), \beta_k(\boldsymbol{x},\boldsymbol{y})$ are invariant with respect to the companion map. This holds true independently of the value of the free parameter which specifies the map. 
	\end{Rem}
\begin{Rem}
	Notice that  $\mathcal{P}_{k}^{(1)} = \beta_k$, and the recurrence formula \eqref{eq:P-k-l-ind} gives
	\begin{equation}
	\mathcal{P}_{k}^{(\ell)} = \mathcal{P}_{k}^{(\ell - 1)}  \beta_{k+\ell - 1} -\mathcal{P}_{k}^{(\ell-2)} \alpha_{k+\ell - 1},
	\end{equation}
	which easily implies invariance of the polynomials $\mathcal{P}_{k}^{(\ell)}$ with respect to the companion map.
\end{Rem}

Let us close this Section with similar discussion of braid relations \eqref{eq:braid-r} in the context of the companion map. Given three sets of indeterminates 
$(\boldsymbol{x},\boldsymbol{y},\boldsymbol{z} )= (x_k,y_k,z_k)_{k\in\mathbb{Z}}$, define (using the same symbols as above but with different number of arguments, the meaning will be clear from the context) functions 
$\alpha_k(\boldsymbol{x},\boldsymbol{y},\boldsymbol{z} )$, $\beta_k(\boldsymbol{x},\boldsymbol{y},\boldsymbol{z} )$, and $\gamma_k(\boldsymbol{x},\boldsymbol{y},\boldsymbol{z} )$, $k\in \mathbb{Z}$, by
\begin{align}
\alpha_k(\boldsymbol{x},\boldsymbol{y},\boldsymbol{z} ) & = x_k y_k z_k , \\
 \beta_k(\boldsymbol{x},\boldsymbol{y},\boldsymbol{z} )  & = y_k z_k + x_{k+1} z_k + x_{k+1} y_{k+1},\\
\gamma_k(\boldsymbol{x},\boldsymbol{y},\boldsymbol{z} ) & = z_k + y_{k+1} + x_ {k+2} .
\end{align}
\begin{Prop}
The above three-argument functions 
$(\alpha_k, \beta_k, \gamma_k)_{k\in \mathbb{Z}} $ are invariant with respect to the maps $r_1 = r \times \mathrm{id}$ and $r_2 = \mathrm{id} \times r$ whose superpositions give the braid relations \eqref{eq:braid-r}.
\end{Prop}
\begin{proof}
	The result follows from the observation that 
\begin{align*}
\alpha_k(\boldsymbol{x},\boldsymbol{y},\boldsymbol{z} ) & = \alpha_k(\boldsymbol{x},\boldsymbol{y}) \, z_k = 
x_k \, \alpha_k(\boldsymbol{y},\boldsymbol{z} ), \\
\beta_k(\boldsymbol{x},\boldsymbol{y},\boldsymbol{z} ) & = \beta_k( \boldsymbol{x},\boldsymbol{y}) z_k + \alpha_{k+1} (\boldsymbol{x},\boldsymbol{y}) =
\alpha_{k} (\boldsymbol{y},\boldsymbol{z}) + x_{k+1} \beta_k(\boldsymbol{y},\boldsymbol{z} )
,\\
\gamma_k(\boldsymbol{x},\boldsymbol{y},\boldsymbol{z} )  & = z_k + \beta_{k+1}(\boldsymbol{x},\boldsymbol{y} ) = \beta_k(\boldsymbol{y},\boldsymbol{z} ) + x_{k+2},
\end{align*}
and the invariance of the two-argument functions with respect to the companion map.  	
\end{proof}

\section{The companion map in the periodic reduction} \label{sec:KP-series} 
\subsection{Explicit formulas for periodic companion map}
In this Section we assume that the arguments $(x_k,y_k)$ of the companion map are periodic with period $K$, i.e. $x_{k+K} = x_k$, 
$y_{k+K} = y_k$. We will be looking for the explicit form of the map which preserves the periodicity condition, i.e. $x_{k+K}^\prime = x_k^\prime$, $y_{k+K}^\prime = y_k^\prime$ as well.  
Ruling out the identity mapping we define $h_k$, as it was done in the commutative case in \cite{NY-RSK}, by 
\begin{equation} \label{eq:x-h}
x_{k+1}^\prime = x_{k+1} + h_k^{-1}.
\end{equation}
Then the second part of the system \eqref{eq:xy-k} gives
\begin{equation} \label{eq:y-h}
y_{k}^\prime + h_k^{-1} = y_{k} ,
\end{equation}
and the first part leads to equations
\begin{equation} \label{eq:hxy-rec}
h_k x_{k+1} + 1 = y_{k+1}h_{k+1} .
\end{equation}

In  order to make subsequent formulas more readable, let us introduce the following notation
\begin{equation*}
[y^{-1}x]_k^0 = 1, \qquad [y^{-1}x]_k^{n+1} = y_{k-1}^{-1}[y^{-1}x]_{k-1}^n x_k \qquad \text{where} \quad n\geq 0.
\end{equation*}
For example
\begin{equation*}
[y^{-1}x]_2^4 = y_1^{-1}y_0^{-1}y_{-1}^{-1} y_{-2}^{-1} x_{-1} x_0 x_1 x_2 
\end{equation*}
which in $K=2$-periodic case should mean
\begin{equation*}
[y^{-1}x]_2^4 = y_1^{-1}y_2^{-1}y_{1}^{-1} y_{2}^{-1}  x_{1} x_2 x_1 x_2,
\end{equation*}
\begin{Prop}
	The formal series solution of the system \eqref{eq:hxy-rec} is given by
	\begin{equation} \label{eq:hk-1}
	h_k = y_k^{-1} \sum_{n\geq 0} [y^{-1}x]_k^n = 
	\sum_{n\geq 0} (y_{k-n} y_{k-n+1} \dots y_k)^{-1} x_{k-n+1} x_{k-n+2} \dots x_k ,
	\end{equation}
	i.e. we regard $h_k$ as Mal'cev--Neumann series in the symbols $y_k^{-1}$, $x_k$ for $k\in\mathbb{Z}_K$.
\end{Prop}
\begin{proof}
	It is convenient to introduce $Z_k = y_k h_k$, and write down the above system in the form
	\begin{equation} \label{eq:Z_k} 
	Z_k = 1 + y_{k-1}^{-1} Z_{k-1} x_k.
	\end{equation}
	The above system can be solved by the successive approximation method \cite{Stanley-2,Salomaa} starting from $Z_k^{(0)}=0$, and proceeding recursively by 
	\begin{equation*}
	Z_k^{(j+1)} = 1 + y_{k-1}^{-1} Z_{k-1}^{(j)} x_{k}, \qquad j\geq 0.
	\end{equation*}
	The series 
	$$\lim_{j\to\infty} Z_k^{(j)} = \sum_{n\geq 0} [y^{-1}x]_k^n $$ 
	is the  fixed point of equation \eqref{eq:Z_k}.
\end{proof}

Our next step is finding the series representing inverse of the series \eqref{eq:hk-1}. We introduce the following standard notation~\cite{Introduction-QSym}. 
By a composition $\alpha$ of $n$, denoted $\alpha \models  n$, we mean a finite sequence $\alpha=(\alpha_1, \alpha_2 , \dots , \alpha_m)$ of positive integers with $n = \alpha_1 + \alpha_2 + \dots + \alpha_m$. The number $m=|\alpha|$ is called the length of the composition. 
Given composition $\alpha \models n$, define 
\begin{equation*}
[y^{-1}x]^\alpha_k = [y^{-1}x]^{\alpha_1}_k \dots [y^{-1}x]^{\alpha_m}_k  \; .
\end{equation*}
\begin{Prop} \label{prop:h-1}
	The inverse of the series
	\begin{equation*}
	h_k= y_k^{-1} \sum_{n\geq 0} [y^{-1}x]_k^n = y_k^{-1} \left(1 + y_{k-1}^{-1}  x_k +
	y_{k-1}^{-1} y_{k-2}^{-1}  x_{k-1}  x_k + \dots \right)
	\end{equation*}
	is given by
	\begin{equation*}
	h_k^{-1} = \left( \sum_{n\geq 0} \sum_{\alpha \models n} (-1)^{|\alpha|} [y^{-1}x]_k^\alpha \right) y_k = \left( 1 - y_{k-1}^{-1}  x_k +  y_{k-1}^{-1}  x_k \, y_{k-1}^{-1}  x_k -
	y_{k-1}^{-1} y_{k-2}^{-1}  x_{k-1}  x_k + \dots \right) y_k.
	\end{equation*}
\end{Prop}
\begin{proof} Let us look for the inverse of the series $Z_k=\sum_{n\geq 0} [y^{-1}x]_k^n$
	in the form $Z_k^{-1}= \sum_{n\geq 0} c_n$, where $c_n$ is homogeneous polynomial of degree $2n$ in the symbols $y_l^{-1}$ and $x_l$. Then equation $Z_k^{-1} Z_k = 1$ implies
	\begin{equation*}
	c_0 = 1, \qquad 
	c_0 [y^{-1}x]^n_k + c_1 [y^{-1}x]^{n-1}_k + \dots + c_n = 0 , \quad n>0 \;.
	\end{equation*}
	The identification
	\begin{equation*}
	c_n = \sum_{\alpha \models n} (-1)^{|\alpha|} [y^{-1}x]^\alpha_k
	\end{equation*}
	follows by separation of the last component $\alpha_m$ from composition $\alpha$
	\begin{equation*}
	\sum_{\alpha \models n} (-1)^{|\alpha|} [y^{-1}x]^\alpha_k = 
	- \sum_{i=1}^n \left( \sum_{\beta \models n-i} (-1)^{|\beta|} [y^{-1}x]^\beta_k \right)	
	[y^{-1}x]^i_k \; .
	\end{equation*}
\end{proof}
\begin{Rem} 
	There is a natural bijection between compositions of $n$ and subsets of $\{1,2,\dots, n-1\}$, see \cite{Introduction-QSym},
	\begin{equation*}
	\alpha = (\alpha_1, \alpha_2 , \dots , \alpha_m) \mapsto 
	\{ \alpha_1, \alpha_1+\alpha_2, \dots, \alpha_1 + \alpha_2 + \dots + \alpha_{n-1} \} ,
	\end{equation*}
	which gives the number $2^{n-1}$ of compositions of $n$.
	We conclude that each homogenous term $c_n$ has $2^{n-1}$ summands, $n>0$.
\end{Rem}
\begin{Rem}
	By separating the first components of compositions we can check directly that the left inverse of $Z_k$ equals its right inverse.
\end{Rem}
Finally we can state the explicit transition formulas.
\begin{Prop} \label{prop:transf-xy}
	The formal series solutions of equations \eqref{eq:xy-k} which define the companion map  are given by
	\begin{align} \label{eq:transf-x}
	x_k^\prime& = x_k +   \sum_{n\geq 0} \sum_{\alpha \models n} (-1)^{|\alpha|} [y^{-1}x]_{k-1}^\alpha  y_{k-1} = x_k  + y_{k-1} - 
	y_{k-2}^{-1}  x_{k-1} y_{k-1}  + \dots   ,\\ \label{eq:transf-y}
	y_k^\prime & = -   \sum_{n > 0} \sum_{\alpha \models n} (-1)^{|\alpha|} [y^{-1}x]_{k}^\alpha  y_{k} = 
	y_{k-1}^{-1}  x_k y_k -  y_{k-1}^{-1}  x_k \, y_{k-1}^{-1}  x_k y_k +
	y_{k-1}^{-1} y_{k-2}^{-1}  x_{k-1}  x_k y_k + \dots \;  .
	\end{align} 
\end{Prop}
\begin{Rem}
	Equations \eqref{eq:x-h}-\eqref{eq:hxy-rec} imply another form of the companion map
	\begin{equation}
	x_k^\prime = h_{k-1}^{-1} y_k h_k, \qquad y_k^\prime = h_{k-1} x_k h_k^{-1} .
	\end{equation}
\end{Rem}
\begin{Rem}
Equations \eqref{eq:hxy-rec} and \eqref{eq:P-k-l-ind} imply
\begin{equation*}
h_k \, x_{k+1} x_{k+2} \dots x_{k+\ell+1} + \mathcal{P}_{k+1}^{(\ell)} = y_{k+1}y_{k+2} \dots y_{k+\ell+1} \, h_{k+\ell+1},
\end{equation*}
which for $\ell=K-1$ under $K$-periodicity assumption leads to following equation satisfied by $h_k$
\begin{equation}
 y_{k+1}y_{k+2} \dots y_{k+K} h_k =   h_k \,	x_{k+1} x_{k+2} \dots x_{k+K}  + \mathcal{P}_{k+1}^{(K-1)}   \,.
\end{equation}
By the successive approximation technique one can obtain the same series solution \eqref{eq:hk-1}. 
	Under additional condition of centrality of the products $x_{k+1} x_{k+2} \dots x_{k+K} $ and $y_{k+1}y_{k+2} \dots y_{k+K}$ the above equation was considered in~ \cite{Doliwa-YB}. Then its solution reads
	\begin{equation} \label{eq:transf-x-y-rat}
	x_k^\prime = \left[\mathcal{P}_{k}^{(K-1)}\right]^{-1} y_k \mathcal{P}_{k+1}^{(K-1)}, \qquad 
	y_k^\prime = \mathcal{P}_{k}^{(K-1)} x_k \left[\mathcal{P}_{k+1}^{(K-1)}\right]^{-1},
	\end{equation}
	and agrees with the corresponding result obtained in the fully commutative case~\cite{KNY-A,KNY-qKP}.
\end{Rem}
\begin{Rem}
	In the periodic reduction equations \eqref{eq:xy-alpha}-\eqref{eq:xy-beta}  and the matrix refactorization problem stated in Corollary~\ref{cor-alpha-beta} read 
	\begin{equation*}
	\left(  \begin{array}{ccccc}  
	x_1 & 0  & \cdots & 0 & 1 \\
	1    & x_{2} & 0 & \cdots & 0 \\
	0  &	1 & x_3  & \ddots &  \vdots  \\
	\vdots & \ddots & \ddots & \ddots &  0   \\
	0 & \cdots & 0 & 1 & x_{K} 
	\end{array} \right) 
	\left(  \begin{array}{ccccc}  
	y_1 & 0  & \cdots & 0 & 1 \\
	1    & y_{2} & 0 & \cdots & 0 \\
	0  &	1 & y_3  & \ddots &  \vdots  \\
	\vdots & \ddots & \ddots & \ddots &  0   \\
	0 & \cdots & 0 & 1 & y_{K} 
	\end{array} \right) 
	=
	\left(  \begin{array}{ccccc}  
	\alpha_1 & 0  & \cdots & 1 & \beta_K \\
	\beta_1    & \alpha_{2} & 0 & \cdots & 1 \\
	1  & \beta_2 & \alpha_3  & \ddots &  \vdots  \\
	\vdots & \ddots & \ddots & \ddots &  0   \\
	0 & \cdots & 1 & \beta_{K-1} & \alpha_{K} 
	\end{array} \right) .
	\end{equation*}	
Observe that every element $x_k$ or $y_k$ satisfies an equation in a form of a periodic continued fraction. For example, elimination of all variables except from $y_1$ gives
\begin{equation*}
y_1 = \beta_1 - \alpha_{2} \left( \beta_{2} - \alpha_{3} \left( \beta_{3} - \dots  - \alpha_{K} \left(  \beta_{K} - \alpha_{K-1} y_{1}^{-1} \right)^{-1} \dots \right)^{-1}\right)^{-1}
\end{equation*}
As it was observed in~\cite{Doliwa-NC-CF}, solution  of the refactorization problem can be interpreted as a non-commutative analog of the Galois theorem for periodic continued fractions. For more information concerning relation of continued fraction theory to integrable systems see for example \cite{Hirota-1983,IDS}, and especially \cite{NY-W-P} in the context of Weyl groups. Non-commutative continued fractions and integrable systems were considered also in~\cite{DiFrancesco}.
\end{Rem}

\subsection{Warning. New ``absurd formula"}
It turns out that if we rewrite the basic system \eqref{eq:hxy-rec} in the form
	\begin{equation}
	h_k = - x_{k+1}^{-1} + y_{k+1}h_{k+1} x_{k+1}^{-1},
	\end{equation}
	then the successive approximation technique gives its unique solution as
	\begin{equation} \label{eq:hk-2}
	h_k = - \sum_{n\geq 0}[yx^{-1}]_k^n x_{k+1}^{-1},
	\end{equation}
	where $[yx^{-1}]_k^n$ is defined recursively by
\begin{equation*}
[yx^{-1}]_k^0 = 1, \qquad [yx^{-1}]_k^{n+1} = y_{k+1}[yx^{-1}]_{k+1}^n x_{k+2}^{-1}, \qquad n\geq 0.
\end{equation*}	
Extending the notation for compositions in the same way as in the previous case, we obtain the solution	
	\begin{align}
	x_k^\prime& =- x_{k}\sum_{n > 0} \sum_{\alpha \models n} (-1)^{|\alpha|} [yx^{-1}]_{k-1}^\alpha   = x_k y_{k} x_{k+1}^{-1}  -
	 x_k y_{k} x_{k+1}^{-1} y_{k} x_{k+1}^{-1} +  x_k y_{k}  y_{k+1} x_{k+2}^{-1} x_{k+1}^{-1} + \dots   ,\\
	y_k^\prime & =  y_k + x_{k+1} \sum_{n \geq 0} \sum_{\alpha \models n} (-1)^{|\alpha|} [yx^{-1}]_{k}^\alpha = y_k + x_{k+1} - x_{k+1} y_{k+1} x_{k+2}^{-1} + \dots \;  ,
	\end{align} 
	which seems to contradict uniqueness of the companion map.
	
	This paradox can be explained as follows. Both expressions \eqref{eq:hk-1} and \eqref{eq:hk-2} for $h_k$ represent Mal'cev-Neumann series in different spaces and cannot be considered simultaneously. A similar situation happens with the following Euler's ``absurd formula''
	\begin{equation*}
	\sum_{n\in\mathbb{Z}}x^n = 0,
	\end{equation*}
	see detailed discussion in \cite{Cartier}.
	From various aspects of the above ``identity" given there, we can pick up the most relevant one in our context. When considered as rational expressions, the two parts of the series (which however again cannot be considered simultaneously on the series level)
	\begin{equation}
\sum_{n\geq 0} x^n = (1-x)^{-1}, 
\qquad \sum_{n < 0} x^n = \sum_{n> 0} (x^{-1})^n = x^{-1} (1-x^{-1})^{-1} ,
	\end{equation}
	summed up give the desired ``result''.
	
	One can, as a by-product of our considerations, add to the list new ``absurd formula"
	\begin{equation*}
	\sum_{n\in\ZZ} x^n y^n = 0,
	\end{equation*}
	where  $x$ and $y$ are non-commuting symbols (for commuting symbols it reduces to the above one). 
Notice that two  Mal'cev--Neumann series
\begin{equation*}
h^{(1)} = \sum_{n\geq 0} x^n y^n , \qquad h^{(2)} = - \sum_{n> 0} (x^{-1})^{n} (y^{-1})^{n},
\end{equation*}
$h^{(1)} \in \mathbb{Q} (( x,y))$, and 
$h^{(2)} \in \mathbb{Q} (( x^{-1},y^{-1}))$
satisfy equations
\begin{equation*}
h^{(1)} = 1 + x h^{(1)} y, \qquad  h^{(2)} = - x^{-1} y^{-1} + x^{-1} h^{(2)} y^{-1} ,
\end{equation*}
which can be easily transformed one into another.

\section{The companion map and context-free languages}
\label{sec:CF} 
The present Section arose from (unsuccessful) attempts to rewrite the transformation formulas \eqref{eq:transf-x}-\eqref{eq:transf-y} in a more compact form as a rational expression. As we mantioned in Section~\ref{sec:KP-series} such formulas \eqref{eq:transf-x-y-rat} have been  previously obtained under the centrality assumption in~\cite{Doliwa-YB} or in the fully commutative case in~\cite{KNY-A,KNY-qKP}. The answer came from the theory of formal languages \cite{HopcroftUllman}, where analogous questions were discussed \cite{Berstel-Reutenauer,Salomaa,KuichSalomaa} earlier. 

To make the paper self-contained let us recall relevant definitions~\cite{HopcroftUllman}. For a set $A$, called an \emph{alphabet}, we denote by $A^*$ the set consisting of all words (finite sequences of letters in $A$ including the empty sequence $\lambda$).
A context-free grammar $G(N,A,R,Y)$ is given by 
\begin{enumerate}
	\item an alphabet $N$ whose elements are called variables or nonterminals,
	\item an alphabet $A$ (disjoint from $N$) whose elements are called terminals,
	\item a finite set $R$ of rewritting rules or productions, each rule being a pair $Y_i \to \gamma$, where $Y_i \in N$ and $\gamma \in (A\cup N)^*$,
	\item an element $Y\in N$ called the initial variable.
\end{enumerate}
For two words $\gamma_1, \gamma_2$ in $(A\cup N)^*$ we write $\gamma_1 \Rightarrow \gamma_2$ iff $\gamma_1 = \delta_1 Y_i \delta_2$, $\gamma_2 = \delta_1 \gamma \delta_2$ and $Y_i \to \gamma$ is in $R$. Let $\stackrel{*}{\Rightarrow}$ be the reflexive transitive closure of $\Rightarrow$, then the language $L(G)$ generated by $G$ is
$$ L(G) = \{ w\in A^* | Y \stackrel{*}{\Rightarrow} w \} . $$ 
A context-free grammar is termed regular iff the righthand side of every production belongs to $AN \cup \lambda$. A language $L$ is referred to as context-free or regular iff $L=L(G)$ for some context-free or regular grammar, respectively.

To simplify formulas, in this Section we use the notation $\bar{x}_k = y_{k-1}^{-1}$. Then the system \eqref{eq:Z_k} reads
\begin{equation} \label{eq:Z_k-new} 
Z_k = 1 +  \bar{x}_k Z_{k-1} x_k, \qquad k=1,\dots, K.
\end{equation}
Let us consider the support $L_k$ of the series $Z_k$ as a subset of the set $\mathcal{X}^*_K$ consisting of all words over the alphabet $\mathcal{X}_K = \{ x_1, \dots ,x_K, \bar{x}_1, \dots \bar{x}_K \}$. 
\begin{Prop}
	The language $L_k$, $k=1,\dots , K$, is generated by the context-free grammar with variables $\{Y_1, Y_2, \dots , Y_K\}$, terminals $\{ x_1, \dots ,x_K, \bar{x}_1, \dots \bar{x}_K \}$, the production rules
	\begin{equation}
	Y_i\to \lambda, \qquad Y_i\to \bar{x}_i Y_{i-1}  x_i ,
	\end{equation}
	and the initial variable $Y_k$.
\end{Prop}
\begin{Ex}
	In the simplest case $K=1$ the unique solution of the equation (we skip the subscripts)
	\begin{equation} \label{eq:ZZ}
	Z = 1 + \bar{x}Z x,
	\end{equation}
	is the series 
	\begin{equation} \label{eq:ZZ-ser}
	Z= \sum_{n=0} \bar{x}^n x^n = 1 + \bar{x} x + \bar{x} \bar{x} x x  + 
	\bar{x} \bar{x} \bar{x} x x x + \dots ,
	\end{equation}
	with the support $
	L = \{ \bar{x}^n x^n \}_{n\geq 0} = \{ \lambda, \bar{x} x , \bar{x} \bar{x} x x ,
	\bar{x} \bar{x} \bar{x} x x x ,\dots  \}$.
The language $L$ is generated by the grammar with variable $Y$, which is also the initial symbol, terminals $\{ \bar{x}, x\}$ and the productions
	\begin{equation}
	Y\to \lambda, \qquad Y\to \bar{x}Y x .
	\end{equation}
\end{Ex}
	The above language is the standard and the simplest example of a context-free language, which is not regular. In the theory of computing machines this means that it can be accepted by a push-down automaton but cannot be accepted by a finite-state automaton (see \cite{HopcroftUllman} or any textbook of the theory of formal languages). This can be shown by the so-called pumping lemma, and the arguments go over to the case of languages $L_k$. By the known relation between regular languages and rational series~\cite{Berstel-Reutenauer,Salomaa,KuichSalomaa} the corresponding solutions $Z_k$ to system \eqref{eq:Z_k-new} cannot be written as  rational expressions. 

Recall that to find transformation formulas we had to calculate inverses of the series $Z_k$, $k=1,\dots , K$. We remark that the inverse $Z^{-1}$ of the series \eqref{eq:ZZ-ser} was calculated in \cite{BRRZ}. By Proposition~\ref{prop:h-1} the support of the series $Z_k^{-1}$ is the language
\begin{equation}
L_k^* = \bigcup_{n\geq 0} \{ [\bar{x}x]_k^\alpha \; | \; \alpha\models n \} ,
\end{equation}
whose elements are all possible concatenations of words of $L_k$.
Notice that elements of languages $L_k^*$ are in bijection with compositions. From that point of view the elements of the language $L$ of the above example can be put in bijection with generators of the algebra of the non-commutative symmetric functions~\cite{NC-sym,Introduction-QSym}. This Hopf algebra is a generalization of the the algebra of symmetric functions~\cite{Macdonald}, which plays pivotal role in the theory of KP hierarchy~\cite{Miwa-Jimbo-Date}. The corresponding extension for $K>1$ of the Hopf algebra of the non-commutative symmetric functions together with its graded dual (quasi-symmetric functions) was the subject of~\cite{Doliwa-QS}.

\section{Conclusion and open problems}
Starting from periodic reduction of the KP-map for non-commuting symbols we have constructed its companion map in terms of Mal'cev--Neumann series which cannot be reduced to rational expressions. By multidimensional consistency of the KP map our formulas give a realization of the $A$-type Coxeter relations. We studied the action of the finite Coxeter group only, leaving for further research the case of affine Coxeter groups and their eventual connection~\cite{NY-W-P} with non-commutative discrete Painlev\'{e} equations~\cite{Dol-qP6,Hasegawa,Kuroki}. 

An interesting output of our research is the relation of the companion to KP map in the periodic reduction to theory of formal languages. The type of series in non-commuting variables, we have encountered in the present research, corresponds to the so called context-free languages and push-down automata. This aspect certainly deserves deeper studies, especially in the spirit of current applications of integrable systems to combinatorics.

\section*{Acknowledgments}
One of authors (A.D.) would like to thank dr.~Aleksandra Ki\'{s}lak-Malinowska for pointing out that the language $\{ a^n b^n \}_{n\geq 0}$ is not regular. The research of A.D. was supported by National Science Centre, Poland, under grant 2015/19/B/ST2/03575 \emph{Discrete integrable systems -- theory and applications}.
The work of M.N. was supported by JSPS Kakenhi Grant (B)15H03626.

\bibliographystyle{amsplain}

\providecommand{\bysame}{\leavevmode\hbox to3em{\hrulefill}\thinspace}

\end{document}